\documentclass[10pt, a4paper]{article}
\usepackage{xcolor}
\usepackage[a4paper, margin=1in]{geometry}
\newenvironment{myproof}[1][Proof]{\noindent\textit{#1:} }{\hfill$\square$}
\newtheorem{thm}{Theorem}
\newtheorem{rem}{Remark}
\usepackage{algorithm}
\usepackage{cite}
\usepackage{amsmath,amssymb,amsfonts}
\usepackage{algorithmic}
\usepackage{graphicx}
\usepackage{textcomp}

\title{A Necessary and Sufficient Condition for Local Synchronization in Nonlinear Oscillator Networks}
\author{Sanjeev Kumar Pandey, Shaunak Sen, Indra Narayan Kar}
\date{}

\begin{document}

\maketitle

\begin{abstract}
Determining conditions on the coupling strength for the synchronization in networks of interconnected oscillators is a challenging problem in nonlinear dynamics.
While sophisticated mathematical methods have been used to derive conditions, these conditions are usually only sufficient and/ or based on numerical methods. We addressed the gap between the sufficient coupling strength and numerical observations using the Lyapunov-Floquet Theory and the Master Stability Function (MSF) framework. We showed that a positive coupling strength is a necessary and sufficient condition for local synchronization in a network of identical oscillators coupled linearly and in a full-state fashion. For partial state coupling, we showed that a positive coupling constant results in an asymptotic contraction of the trajectories in the state space, which results in synchronisation for two-dimensional oscillators. We extended the results to networks with non-identical coupling over directed graphs and showed that positive coupling constants are sufficient condition for synchronisation. These theoretical results are validated using numerical simulations and experimental implementations. Our results contribute to bridging the gap between the theoretically derived sufficient coupling strengths and the numerically observed ones.
\end{abstract}


\section{Introduction}
\label{sec:introduction}
Synchronization of oscillators is a fundamental phenomenon observed in a wide range of natural and engineered systems \cite{strogatz2004sync,pikovsky2003synchronization,strogatz1996nonlinear,izhikevich2007dynamical}. Representative examples include synchronization in fireflies \cite{ermentrout1991adaptive}, cardiac pacemakers \cite{nunez2016synchronization}, neuronal networks \cite{brody2003simple,kopell2000we,singer1999neuronal}, electronic oscillators \cite{carroll1993synchronizing,tang1983synchronization,wu1995synchronization,liu2019synchronization,buldu2007electronic}, electric power systems \cite{ajala2021robust}, and coordinated robotic systems \cite{sinha2023coupled}. A central challenge in the analysis and design of such systems is to characterize the coupling conditions under which synchronization emerges.

A major advance in the study of local synchronization was the Master Stability Function framework introduced in \cite{pecora1998master}. By linearizing the dynamics about the synchronous solution and analyzing the associated Floquet exponents, the MSF provides a powerful tool for assessing the local stability of synchronization. This framework has been widely applied to various oscillator networks \cite{boccaletti2006complex, braga2024selecting}. However, the MSF approach typically requires numerical computation of Lyapunov or Floquet exponents, and the resulting synchronization conditions are often implicit or computational in nature \cite{shafi2013synchronization}.

Several analytical approaches have been developed to obtain verifiable synchronization conditions, including Lyapunov-based methods, contraction theory, and passivity-based analysis \cite{arcak2011certifying,russo2010global,scardovi2008synchronization,barahona2002synchronization,jadbabaie2004stability,joshi2022synchronization}. While these methods yield rigorous guarantees, they generally provide only sufficient conditions and lead to conservative coupling thresholds. For instance, in coupled Van der Pol oscillators, existing analyses predict synchronization only when the coupling gain exceeds a strictly positive critical value that depends on system parameters and network topology \cite{russo2010global,joshi2022synchronization}.

In contrast, numerical simulations and experimental observations consistently indicate that synchronization can occur for arbitrarily small positive coupling gains. This discrepancy reveals a gap between analytically derived sufficient conditions, typically of the form $K > K^* > 0$, and observed behavior, where synchronization is achieved for $K > 0$ (Fig.~\ref{remark_Kgreater_zero}).

\begin{figure}[h!]
\centering
\includegraphics[width=0.64\textwidth]{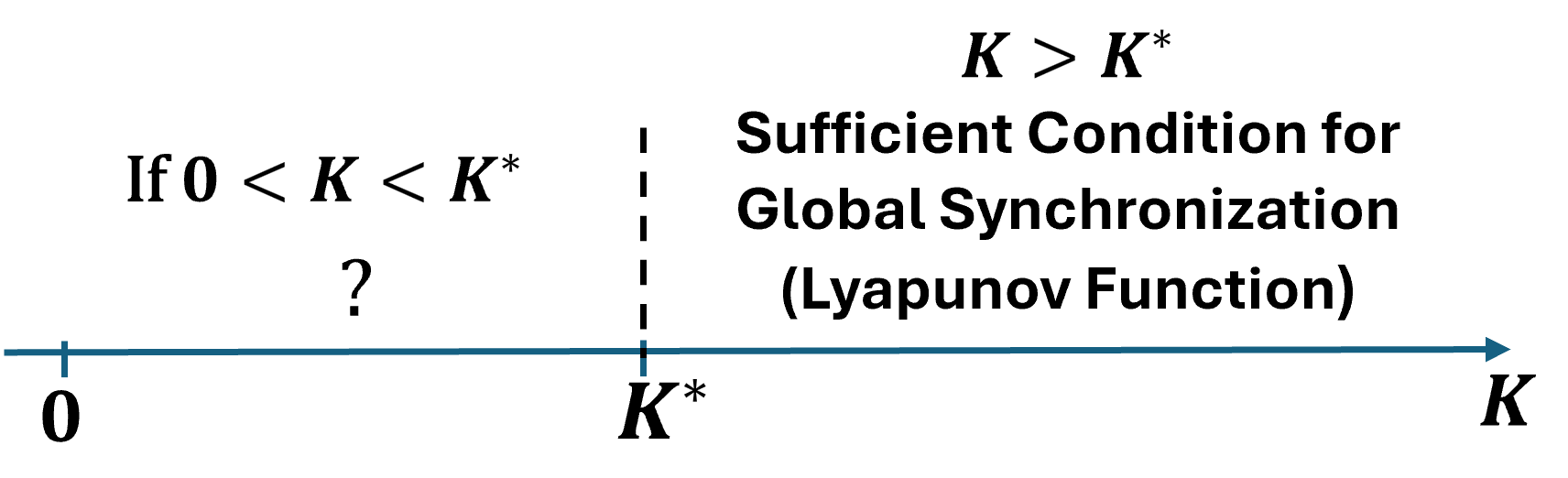}
\caption{Illustration of the gap between analytically derived sufficient conditions ($K>K^*>0$) and numerical observations, which suggest synchronization for $K>0$.}
\label{remark_Kgreater_zero}
\end{figure}

Related Floquet-based analyses have been successfully employed to study synchronization in diffusively coupled and spatially distributed systems, including reaction--diffusion PDEs and compartmental ODE models \cite{shafi2013synchronization,arcak2011certifying}. These works focus on estimating diffusion or coupling strength thresholds that guarantee stability of synchronous limit cycles in infinite- or high-dimensional settings. In contrast, the present work considers finite-dimensional networks of nonlinear oscillators with directed and non-identical coupling gains. In this setting, the coupling matrix is generally non-symmetric and therefore not diagonalizable by an orthogonal transformation. As a result, the associated coupling eigenmodes are not orthogonal, which prevents the standard modal decoupling used in classical MSF analysis and necessitates a different analytical treatment.

Motivated by this gap, we revisit the synchronization problem using Floquet theory, with particular emphasis on the Lyapunov-Floquet transformation. Our objective is not to introduce new Floquet tools, but to exploit their structure to derive tighter analytical conditions on coupling gains. We show that, for identical oscillators coupled linearly in a full-state fashion, a positive coupling constant ($K>0$) is both necessary and sufficient for local synchronization. For partial-state coupling, we establish asymptotic contraction of the state-space volume, which guarantees synchronization for two-dimensional oscillators. We further extend the analysis to networks with non-identical coupling over directed graphs and derive sufficient synchronization conditions in this more general setting. Numerical simulations and experimental results on coupled electronic oscillators validate the theoretical findings.

\section{Mathematical Background}
This section provides a brief overview of the relevant mathematical theory that is used in this study. 
Consider a nonlinear dynamical system
\small{
\begin{equation}
\dot{x}=f\left(x\right), \text{ } x \in \mathbb{R}^{m},
\end{equation}}
\normalsize
where $f(x)$ is a smooth nonlinear function, and the system exhibits a stable limit cycle $x_s(t)$ with time period $T$. Linearizing this system around the limit cycle $x_s(t)$ yields the linear time-periodic system
\small{\begin{equation}\label{linearize}
	\dot{y}=A(t) y, \text{ }y \in \mathbb{R}^{m},
\end{equation}}
\normalsize
where $A(t)=D f\left(x_s(t)\right)$ is the Jacobian of $f$ evaluated along the limit cycle $x_s(t)$. 
The state transition matrix $\phi\left(t, t_0\right)$ describes the evolution of the system from an initial time $t_0$ to a later time $t$. 
For the linearized system (\ref{linearize}), the state transition matrix satisfies
\small{$$
\frac{\partial \phi\left(t, t_0\right)}{\partial t}=A(t) \phi\left(t, t_0\right), \quad \phi\left(t_0, t_0\right)=I
$$}
\normalsize
where $I$ is the identity matrix. 
The fundamental matrix $\phi(T, 0)$, which maps the state at time $t=0$ to $t=T$, plays a central role in Floquet theory.

\subsection{Floquet Theory}\label{A}
Floquet theory was developed to describe the behavior of a system of linear differential equations with time-periodic coefficients \cite{iakubovich_linear_1975}. 
The theory introduces the concept of the Floquet multipliers ($\mu$), which are eigenvalue-like entities that describe the exponential growth or decay of solutions over one time period of the system. 
The fundamental solutions denote the state transition by $\phi(t,t_{0})$. Define $R$ by $e^{RT}=\phi(T,0)$ and $P^{-1}(t)=\phi(t,0)e^{-RT}$. 
Then $P^{-1}(t)$ is periodic with period $T$ as $P^{-1}(t+T)=\phi(t+T,0)e^{ -R(t+T)}=\phi(t+T,T)\phi(T,0)e^{-RT}e^{-Rt}=\phi(t,0)e^{-Rt}=P^{-1}(t)$. 
In the co-ordinates $z=P(t)y$, the system is the linear time invariant system $\dot z=Rz$ as $\dot P(t)P^{-1}(t)+P(t)A(t)P^{-1}(t)=R.$ 
The eigenvalues of $R$, called the Floquet exponents, determine the stability. 
Floquet multipliers are the eigenvalues of $e^{RT}$. 
For asymptotic stability, one of the Floquet multiplier has absolute value as $1$, representing perturbations along the limit cycle, and others have moduli strictly less than $1$.

\subsection{Master Stability Function}\label{B} 
Consider $n$ identical oscillators
\begin{equation} \label{uncoupled}
	\dot{x}_{i}=f\left(x_{i}\right), \text{ } x_{i} \in \mathbb{R}^{m}, \text{ } i \in \{1,2,...,n\},
\end{equation}
diffusively and identically coupled in a fully connected network \cite{pecora1990synchronization,boccaletti2006complex}
\begin{equation} \label{eq:(1)}
	\dot{x}_{i}=f\left(x_{i}\right)+K \sum_{j\in {N_{i}}}G_{ij} H\left(x_{j}\right),
\end{equation}
where $K$ is the coupling constant, $N_{i}$ is the neighborhood of the oscillator $i$, $G$ is the graph Laplacian and $H\left(x_j\right)$ is a coupling function that determines how variables of oscillator $j$ influence oscillator $i$. 
For linear coupling, the coupling between an oscillator $i$ and an oscillator $j \in N_{i}$ is $K(x_{j}-x_{i})$. 
Linearizing the system around the synchronous state  ($x_{s}$) results in the variational equation
\begin{equation} \label{var_equation}
	\dot {Y}=\left[I_{n} \otimes Df\left(x_s\right) - KG \otimes DH\left(x_s\right)\right] Y,
\end{equation}
where $Y=\left[y_1^T, y_2^T, \ldots, y_n^T\right]^T$ is a $nm$-dimensional state vector, and $\otimes$ is the Kronecker product. $D H\left(x_s\right)$ is the Jacobian of $H$ evaluated at $x_s(t)$. 
Equation (\ref{var_equation}) can be written in block matrix form by introducing a transformation $\zeta=(Q \otimes I_{n})Y$, where $Q$ is the transformation matrix. On diagonalising $G$ through the transformation $Q$, where $Q^{-1}GQ= \text{diag} \left[0, \lambda_{2}, \lambda_{3},..., \lambda_{n}\right]$, $0< \lambda_{2} \leq \lambda_{3}\leq,...,\leq \lambda_{n}$. 
In transformed coordinates, the dynamics are in block form, with each block representing an eigenmode
\begin{equation}\label{gamma_eq}
	\dot{\zeta_{i}}=[DF\left(x_s\right)-K\lambda_{i} DH\left(x_s\right)] \zeta_{i}, \text{ } i = 1,2,...,n.
\end{equation}
$\zeta_{i}$ is referred as the synchronization mode. 
The dynamics corresponding to the first eigenmode is $\dot{\zeta}_{1}=DF(x_S) \zeta_{1}$, which is the same as the linearised dynamics around the limit cycle. 
Assuming this mode to be stable, the stability of the synchronized state is determined by the Floquet multipliers of all other eigenmodes. 
The Master Stability Function ($\mu_{max}(\lambda)$) is the largest non-unity Floquet multiplier of the system matrix $[ DF\left(x_s\right)-K\lambda DH\left(x_s\right)]$. 
The necessary and sufficient condition for the local stability is $\mu_{max}(\lambda)<1$.

\section{CONDITIONS FOR SYNCHRONIZATION}
This section investigates the conditions for synchronization of identical oscillators coupled identically and linearly. Synchronization of two or more identical oscillators is defined as the situation when the difference between their corresponding states is asymptotically zero. 

\begin{thm}\label{Theorem1}
A network of identical oscillators (\ref{eq:(1)}) coupled identically and linearly in full-state fashion synchronizes if and only if $K>0$.
\end{thm}
\begin{myproof}
The proof relies on the Lyapunov-Floquet transformation $P(t) \in \mathbb{R}^{m \times m}$ \cite{sinha1996liapunov} and the fact that for identical linear full-state coupling, $DH=I.$ 
Applying the transformation $P(t)$ at each eigenmode $\lambda_{i}$ (refer (\ref{gamma_eq})) results in
\begin{equation}
\dot Z_{i}=(R-KP(t)DH(x_{s})P^{-1}(t)) Z_{i},\text{ } i = 1,2,...,n,
\end{equation}
\begin{equation}\label{lya-flo}
\dot Z_{i}=(R-K\lambda_{i}I) Z_{i},\text{ } i = 1,2,...,n,
\end{equation}
where $R$ is defined in subsection \ref{A}. For $K=0$, system dynamics (\ref{gamma_eq}) become uncoupled, and the matrix $R \in \mathbb{R}^{m \times m}$ has one eigenvalue at $0$ and other $(m-1)$ eigenvalues with strictly negative real parts. 
For coupled dynamics in (\ref{lya-flo}), the eigenvalues are the roots of the polynomial,
$$\det(sI-(R-K\lambda_{i}I))=0, \Rightarrow \det((s+K\lambda_{i})I-R)=0.$$
The roots of the above characteristic equation depend on the coupling strength $K$ and the eigenmode $\lambda_{i}$. 
To prove sufficiency, first assume that $K>0$ and define $s+K\lambda_{i}=\alpha \text{ } \Rightarrow s=\alpha-K\lambda_{i}$. For $K=0, \text{ } s=\alpha$, and the eigenvalues are same as those of $R$. 
One can conclude that the set of eigenvalues of $(R-K\lambda_{i}I)$ are the same as the set of the eigenvalues of $R$ except offset by $-K\lambda_{i}$. 
All eigenvalues of $R$ corresponding to each synchronization mode ($\lambda_{i}$) have a strictly negative real part except for one eigenvalue of the first eigenmode ($\lambda_{1}=0$), which is zero.
	
To prove the necessary condition, we suppose that $K<0$. Define $s+K\lambda_{i}=\alpha \text{ }\Rightarrow s=\alpha-K\lambda_{i}$. 
We use the method of contradiction and assume that the synchronization is achieved if $K<0$. 
The matrix $R \in \mathbb{R}^{m \times m}$ has one eigenvalue $0$ and other $(m-1)$ eigenvalues with strictly negative real parts. 
However, as $K<0$, the first eigenvalue of matrix $R$, which is zero, moves to the right half of the complex plane for all eigenmode $\lambda_{i}$, ($i=2,...,n$). 
It means that at least one of the eigenvalues of ($R-K\lambda_{i}I$) corresponding to $i^{th}$ eigenmode has a positive real part. 
This contradicts the supposition. 
\end{myproof}

\begin{rem}
Using the Lyapunov-Floquet transformation, the necessary and sufficient condition on the coupling gain $K$ for the synchronization of $n$ oscillators has been obtained as $K>0$. 
This condition is both necessary and sufficient, distinguishing it from many existing approaches that primarily offer only sufficient conditions and often rely on numerical methods. 

For example, a contraction theory-based approach provides a sufficient condition for the synchronization of $n$ coupled oscillators as $K>\frac{\alpha}{n}$, where $\alpha$ is an upper bound of a matrix measure of the associated Jacobian \cite{russo2010global}. 
Similarly, a Lyapunov function-based approach yields a sufficient condition for the synchronization of $n$ coupled Van der Pol oscillators as $K>\mu/\lambda_{2}(L)$, where $\mu$ is the system parameter \cite{joshi2022synchronization}. 
While these approaches can be conservative, the result derived here offers a broader perspective with a balance of analytical rigor and practical applicability. 
This non-conservativeness makes this findings more general and robust in addressing synchronization phenomena across benchmark oscillator networks.
\end{rem}

Moreover our results  strengthen the existing results \cite{shafi2013synchronization} by providing a condition for synchronization solely based on positive coupling gain, rather than relying on threshold-based sufficient conditions that requires numerical computations.

Theorem \ref{Theorem1} is derived in the case that all states are coupled. 
When only some states are coupled (referred to as a partially coupled system). 
In general, $DH=diag[d_{1},d_{2},...,d_{n}]$ where $d_{i}\in \{0,1\}$ $i=1,2,...,n$ and not all $d_{i}=1$ simultaneously. 
The above proof doesn't work because the matrix multiplication of $P(t)$ and $DH$ is not necessarily commutative. 
However, a sufficient condition on the time evolution of the determinant of the state transition matrix is obtained using Abel-Jacobi-Liouville (AJL) identity \cite{brockett2015finite}.

\begin{thm}\label{lemma_partial}
For the network of identical oscillators (\ref{eq:(1)}) coupled linearly in a partial state fashion, $K>0$ implies $ \lim\limits_{t \to \infty}\operatorname{det} \phi(t,t_{0})=0$, where $\phi(t,t_{0})$ is the state transition matrix of (\ref{gamma_eq}).
\end{thm}
\begin{myproof}
The AJL identity for an uncoupled system (\ref{linearize}) is 
\begin{equation}
\operatorname{det} \phi\left(t,t_{0}\right)=\exp\left[ \int_{t_{0}}^{ t}{ \operatorname{tr} (DF(x_{s}(\tau))d\tau)}\right],
\end{equation}
where $\operatorname{det} \phi\left(t,t_{0}\right)$ is the state transition matrix.
The determinant of the transition matrix can be interpreted as a measure of the volume in the phase space. 
For the linearized partial state coupled system (\ref{gamma_eq}), the state transition matrix corresponding to each eigenmode follows
\begin{equation*}\label{STM_For}
\operatorname{det} \phi\left(t,t_{0}\right)=\exp\left[ \int_{t_{0}}^{t}{ \operatorname{tr} (DF(x_{s}(\tau))-K \lambda_{i}DH(x_{s}))d\tau}\right], 
\end{equation*}
\begin{equation}\label{state_tra_sim}
=\exp\left[ \int_{t_{0}}^{t}{(\operatorname{tr} (DF(x_{s}(\tau)))d\tau)}\right]\exp(-K \lambda_{i}(t-t_{0})).
\end{equation}
For positive coupling gain \( K > 0 \), equation (\ref{state_tra_sim}) shows that the determinant of the state transition matrix decays exponentially over time. 
As a result, \( \det \phi(t, t_0) \rightarrow 0 \) as \( t \rightarrow \infty \). 
This implies that the volume in the phase space shrinks, indicating that system trajectories converge, which supports synchronization. 
\end{myproof}

\begin{figure}[!h]
\centering
\includegraphics[width=8cm, height=4cm]{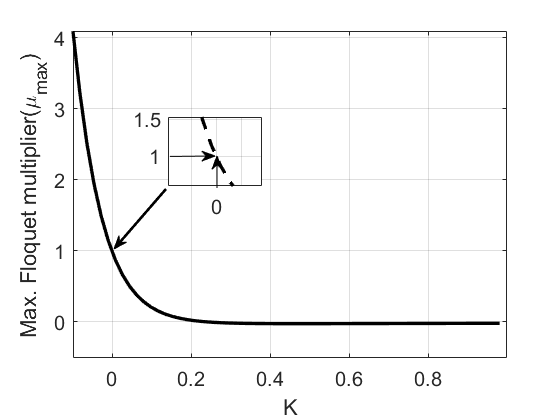}
\caption{Numerical computation of the Master Stability Function for a coupled Van der Pol oscillator shows that the maximum Floquet multiplier decreases as the coupling strength ($K$) increases.}
\label{fig:MSF_Plot}
\end{figure}

We further used a numerical approach to study the synchronization behavior of a second-order partial-state coupled system (\ref{gamma_eq}). 
For this purpose, the Master Stability Function was calculated numerically \cite{pecora1998master}. 
We found the maximum Floquet multiplier is a decreasing function of the coupling strength ($K$). 
This has been shown in Fig. \ref{fig:MSF_Plot}. 
As the maximum Floquet multiplier is less than one for $K>0$, the computation suggests synchronization is achieved.
\section{NON-IDENTICAL COUPLING}
When the oscillators are coupled identically, a single scalar constant $K$ is sufficient to characterize the coupling strength. However, when the oscillators are coupled non-identically, then different coupling constants $k_{ij}$ are needed and the above theoretical framework does not suffice.
We considered the oscillators as defined in (\ref{uncoupled}) coupled diffusively over a fully connected network, where \( k_{ij} \) represents the effective coupling strength between oscillator \( i \) and oscillator \( j \).
In a fully connected network, \(k_{ij}\) also accounts for the network topology, effectively absorbing the Laplacian structure.
To analyze the stability of the coupled system, we linearize the dynamics around the synchronous state \( x_s \), where \( Df(x_s) \) represents the Jacobian of the internal oscillator dynamics and \( DH(x_s) \) represents the Jacobian of the interaction function.
The resulting linearized system is given by 
\begin{equation}\label{coupled_directed_reduced} 
\delta \dot{X} = (I_n \otimes Df(x_s)) \delta X + (\mathbf{K} \otimes DH(x_s)) \delta X, 
\end{equation} 
where \( \delta x_i = x_i - x_s \) represents the deviation of oscillator \( i \) from the synchronous state.
We express the system in matrix form by defining \( X = [x_1, x_2, \ldots, x_n]^T \) as the \( nm \)-dimensional state vector, with \( X_s = [x_s, x_s, \ldots, x_s]^T \) representing the synchronized state, and \( \delta X = [\delta x_1, \delta x_2, \ldots, \delta x_n]^T \) representing the vector of deviations of all oscillators from synchronization.
Here, \( \otimes \) denotes the Kronecker product and \( I_n \) is the identity matrix of size \( n \).
The matrix \( \mathbf{K} \) can be interpreted as the Laplacian-like coupling matrix, which encapsulates the coupling strengths \( k_{ij} \) for all pairs of oscillators in the network. It is defined as 
$\mathbf{K} = [k_{ij}]$, where $k_{ij}$ represents the coupling strength from node $j$ to node $i$ for $ i \ne j$.
The diagonal entries are chosen such that the row sums vanish, $k_{ii} = -\sum_{j \ne i} k_{ij}$.
This structure ensures the total incoming coupling strength to each node is balanced by its diagonal entry, a common feature in the analysis of synchronization dynamics. For example, in the matrix ($\mathbf{K}$), the off-diagonal terms correspond to the pairwise coupling strengths, while each diagonal element is defined such that the sum of each row is zero, i.e., \( \sum_j k_{ij} = 0 \) for each \( i \). This form guarantees that the synchronous solution \( x_1 = x_2 = \cdots = x_n = x_s \) remains invariant under the dynamics. 

We simplified the linearized system using properties of the Kronecker product. Consider the Kronecker product of the coupling matrix \( \mathbf{K} \) and the Jacobian matrix \( DH(x_s) \). To facilitate diagonalization, we introduced a transformation \( \delta X = T \delta \tilde{X} \) with \( T = V \otimes I_m \), where \( V \) diagonalizes \( \mathbf{K} \). With this transformation, the system becomes 

\[ 
T \delta \dot{\tilde{X}} = \left(I_n \otimes Df(x_s) - \mathbf{K} \otimes DH(x_s)\right) T \delta \tilde{X}. 
\] 

Premultiplying both sides by \( T^{-1} \) yields 

\[ 
\delta \dot{\tilde{X}} = T^{-1} \left(I_n \otimes Df(x_s) - \mathbf{K} \otimes DH(x_s)\right) T \delta \tilde{X}. 
\] 

Using the Kronecker product property \( (A \otimes B)(C \otimes D)(E \otimes F) = (ACE) \otimes (BDF) \), the system simplifies to 

\[ 
\delta \dot{\tilde{X}} = \left(I_n \otimes Df(x_s) - Q \otimes DH(x_s)\right) \delta \tilde{X}, 
\] 

where \( Q = V^{-1} \mathbf{K} V \) is a diagonal matrix containing the eigenvalues \( \lambda_i \) of \( \mathbf{K} \). The decoupled system is 
\begin{equation}\label{coupled_directed_final_reduced} 
\delta \dot{\tilde{X}}_i = (Df(x_s) - \lambda_i DH(x_s)) \delta \tilde{X}_i, \quad i = 1, \ldots, n, 
\end{equation} 

where \( \delta \tilde{X}_i \) denotes the transformed perturbation for the \( i \)-th eigenmode. In this framework, synchronization analysis is reduced to verifying the stability of each decoupled mode. As established in Theorem~\ref{3_node}, synchronization is achieved if one eigenvalue of \( \mathbf{K} \) is zero (corresponding to the synchronization manifold), and all remaining eigenvalues have positive real parts. 

\begin{thm}\label{3_node}
A network of identical oscillators coupled non-identically and linearly over a directed graph in a full-state fashion \ref{eq:(1)} synchronizes if one of the eigenvalues of the coupling gain matrix ($\mathbf{K}$) is zero and the remaining have a positive real part.
\end{thm}

\begin{myproof}
The proof relies on the Lyapunov-Floquet transformation $P(t) \in \mathbb{R}^{m \times m}$ (\ref{A}). 
Applying the transformation $P(t)$ at each eigenmode $\lambda_{i}$ (refer (\ref{coupled_directed_final_reduced})) results in
\begin{equation}
\dot Z_{i}=(R-P(t) \eta P^{-1}(t)) Z_{i},\text{ } i = 1,\ldots,n \text{ \& } \eta=\lambda_{i}DH(x_{s}),
\end{equation}
\begin{equation}\label{lya-floq}
\dot Z_{i}=(R-\lambda_{i}I) Z_{i},\text{ } i = 1,2,...,n,
\end{equation}
When coupling strength $\mathbf{K}=0$, system dynamics are uncoupled, and the matrix $R \in \mathbb{R}^{m \times m}$ has one eigenvalue at $0$ and other $(m-1)$ eigenvalues with strictly negative real parts. 
For coupled dynamics ($K_{1},K_{2}, \ldots, K_{n} \neq 0$) in (\ref{lya-floq}), the eigenvalues are the roots of the polynomial,
$$\det(sI-(R-\lambda_{i}I))=0,$$
$$ \Rightarrow \det((s+\lambda_{i})IR)=0.$$
The roots of the above characteristic equation are contingent upon the coupling strength and the corresponding eigenmode $\lambda_{i}$. 
Define a new variable $\alpha = s + \lambda_{i}$. 
When $\mathbf{K} = 0$, $s = \alpha$, and the eigenvalues are identical to those of matrix $R$. 
Consequently, the eigenvalues of $(R-\lambda_{i}I)$ are equivalent to those of $R$, but with an offset of $(-\lambda_{i})$. 
All eigenvalues of $R$ associated with each synchronization mode possess negative real parts, except for the first eigenmode $(\lambda_{1}=0)$, which is zero. 
\end{myproof}

\begin{rem} 
The sufficient condition derived above using the Lyapunov-Floquet transformation ensures that the synchronized state is locally exponentially stable. Specifically, for non-identical but positive coupling gains $\mathbf{K}$, the shifted eigenvalues $\left(R-\lambda_i I\right)$ have negative real parts for all transverse modes $\left(\lambda_i \neq 0\right)$, which guarantees convergence towards the synchronization manifold. Thus, the synchronized state is stable under the above-stated sufficient condition, but the necessary condition may not hold because the coupling gain matrix is not symmetric.
\end{rem}

\begin{rem}
To establish the validity of Theorem~\ref{3_node}, the following assumptions and conditions are essential:
\begin{enumerate}
    \item The coupled network must be fully connected over a directed graph. The coupling structure is represented by  \(k_{ij}\), this is a local coupling gain between node \( j \) and node \( i \), and should not be confused with the network coupling matrix \( \mathbf{K} \in \mathbb{R}^{n \times n} \) used in the compact representation of the full system. Coupling gain \( \mathbf{K}\) can be interpreted as a Laplacian matrix.
    
    \item The eigenvalue structure of the coupling matrix \( \mathbf{K} \) must satisfy the condition that one eigenvalue is zero, while all remaining eigenvalues possess strictly positive real parts. This ensures that the synchronization manifold is invariant under the system dynamics and locally exponentially stable.
    
    \item The applicability of this theorem extends beyond idealized settings and holds in generalized graph topologies. In particular, when all coupling constants $k_{i j}>0$, the Laplacian-like structure of the coupling matrix $\mathbf{K}$, together with Gershgorin's Theorem, ensures that $\mathbf{K}$ has one eigenvalue at zero and the remaining with strictly positive real parts. This spectral property satisfies the sufficient condition for synchronization, thereby demonstrating the practical relevance of the result.
    \end{enumerate}
\end{rem}
We illustrated the above condition through a general network example.
 \vspace{-0.35cm}
\begin{figure}[h]
\centering
\includegraphics[width=0.35\textwidth]{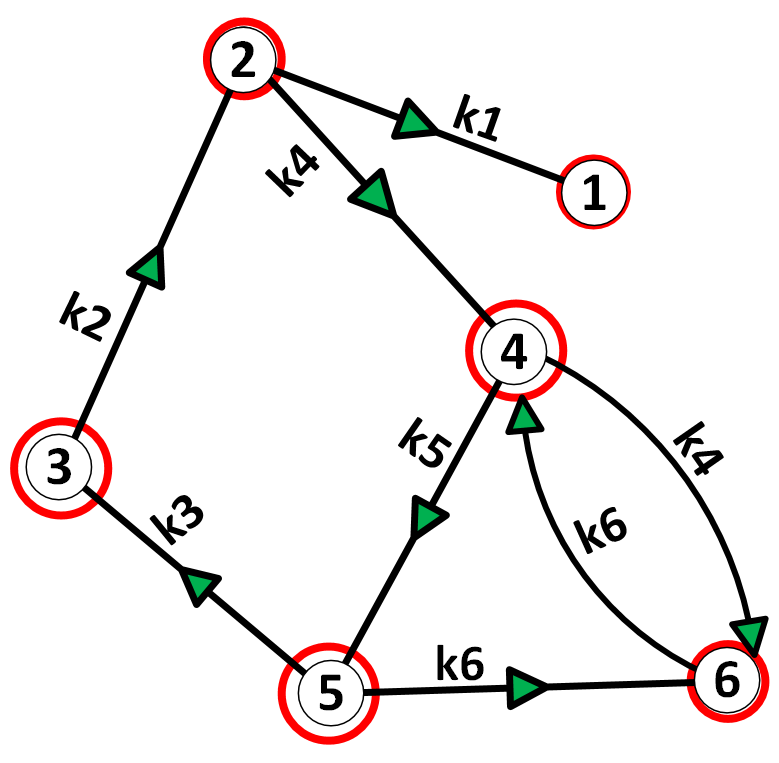}
\caption{Six node directed graph with coupling gain $\mathbf{k}_{ij}$ as follows $K_1=\mathbf{k}_{21},$ $K_2=\mathbf{k}_{32},$ $K_3=\mathbf{k}_{53},$ $K_4=\mathbf{k}_{24},$ $K_5=\mathbf{k}_{45},$ and $K_6=\mathbf{k}_{56}$}
\label{hetro_K_gen}
\end{figure}

\subsection*{Generalized Graph Network}
Consider a fully connected six-node graph is chosen to study the synchronization of non-identical coupled oscillators.
This study explores the synchronization conditions for a system of identical oscillators coupled with varying positive weights over a directed graph Fig. \ref{hetro_K_gen}.
Such graph studies help to understand a more realistic and complex topology often observed in real-world systems. 

Figure \ref{hetro_K_gen} illustrates such a network, where six nodes are interconnected through directed edges. \( \mathbf{K} \) is the coupling matrix

{\small
\addtolength{\arraycolsep}{-3pt} 
\begin{equation} \label{eq1}
\begin{split}
\mathbf{K}=\left[\begin{array}{cccccc}
K_{1} & -K_{1} & 0 & 0 & 0 & 0 \\
0 & K_{2} & -K_{2} & 0 & 0 & 0 \\
0 & 0 & K_{3} & 0 & -K_{3} & 0 \\
0 & -K_{4} & 0 & 2K_{4} & 0 & -K_{4} \\
0 & 0 & 0 & -K_{5} & K_{5} & 0 \\
0 & 0 & 0 & -K_{6} & -K_{6} & 2K_{6} \\
\end{array}\right].
\end{split}
\end{equation}
}

In this six node directed graph structure, each block in the coupling matrix corresponds to $\mathbf{k}_{ij}$ as follows $K_1=\mathbf{k}_{21},$ $K_2=\mathbf{k}_{32},$ $K_3=\mathbf{k}_{53},$ $K_4=\mathbf{k}_{24},$ $K_5=\mathbf{k}_{45},$ and $K_6=\mathbf{k}_{56}$.
We examined conditions, the matrix $\mathbf{K}$ for it to have eigenvalues that satisfy the condition given in Theorem \ref{3_node}.
For this verification, the Gershgorin circle theorem approaches have been adopted and used to analyze the eigenvalue distribution of the $\mathbf{K}$ matrix (\ref{eq1}).
This analysis showed that all eigenvalues lie within their corresponding Gershgorin disks and touch the origin, indicating that the eigenvalues lie in the right-half plane, with no repeated eigenvalues.
Notably, since the underlying graph is connected and the coupling matrix exhibits a Laplacian-like structure, it follows from spectral graph theory that the eigenvalue at zero is simple, and its associated eigenspace is one-dimensional.
This eigenspace is spanned by the constant vector \( [1, 1, \ldots, 1]^T \), which represents the synchronization manifold where all oscillators evolve identically. This confirms the applicability and validity of Theorem~\ref{3_node}.
\begin{figure}[!h]
\centering
\includegraphics[width=0.5\textwidth]{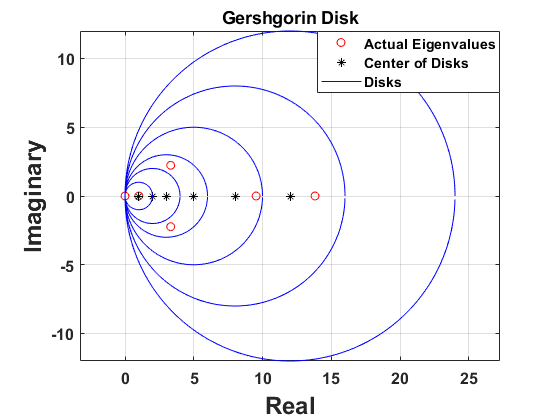}
\caption{Existence of positive real non-repeating roots for nonidentical coupling strength $K_{1}=1, K_{2}=2, K_{3}=3, K_{4}=4, K_{5}=5$ and $K_{6}=6$.}
 \label{Gresh_6node}
\end{figure}

Further analysis of the eigenvalue distribution (Fig. \ref{Gresh_6node}) revealed important insights into the behavior of the system. 
It is also worth noting that none of the eigenvalues are repeated, which implies that the system does not exhibit any degeneracy in its eigenvalue spectrum.
These observations lead to the important conclusion, synchronization of the coupled system can be achieved if the $K>0$ conditions (Theorem \ref{3_node}) are satisfied.
In other words, the system’s ability to synchronize is guaranteed when the eigenvalues of the coupling matrix lie in the right-half complex plane, with one eigenvalue at zero and the rest having strictly positive real parts. This spectral property arises, for example, when $\mathbf{K > 0}$, as the Gershgorin disks are positioned such that they touch the origin and lie entirely within the right-half plane. Thus, this structure provides a sufficient condition for synchronization.

\textbf{Numerical simulation:} 
We numerically simulated the six-node directed generalized graph \ref{hetro_K_gen} for coupled Van der Pol oscillators to further illustrated the dynamics of complete synchronization under nonidentical coupling. 
Initially, for $t \in [0, 15]$ sec, the oscillators evolve independently, exhibiting distinct trajectories due to their unique initial conditions. 
At $t = 15$ sec, the coupling is activated, leading to a transient phase where the oscillators adjust to the directed interactions imposed by the graph topology.
 \vspace{-0.5cm}
\begin{figure}[h!]
\centering
\includegraphics[width=0.5\textwidth]{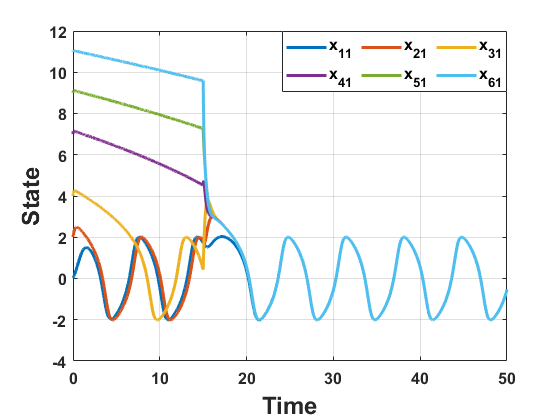}
\caption{Numerical simulation of coupled Van der Pol oscillators over a directed graph for $\mu=1$. Three node case (initial condition: $(x_{10}=[1,2],x_{20}=[3,4],x_{30}=[5,6],x_{40}=[7,8]$ $x_{50}=[9,10], x_{60}=[11,12]$ ) and gain ($K_{1}=1, K_{2}=2, K_{3}=3, K_{4}=4, K_{5}=5$ and $K_{6}=6$).}
 \label{6_node_genralized_ode}
\end{figure}

Following the activation of coupling, the system transitions into a state of complete synchronization. 
All oscillators converge to a common periodic trajectory, irrespective of their initial states or positions within the directed network. 
This result provides a sufficient condition for nonidentical coupling and directed topologies to induce synchronized behavior in networked nonlinear oscillators, provided that the coupling strength and network structure are appropriately configured.
Analyzing this generalized graph offers valuable insights into synchronization phenomena in networks with diverse connection patterns and non-uniform interaction strengths. 
The observed synchronization highlights the influence of directed coupling in achieving coherent dynamics in complex systems.

\section{Experimental Result: Partial-state coupling}
To complement the numerical simulations, an electronic testbed was designed to investigate the synchronization behavior of nonlinear oscillators.

\begin{table}[!h]
\caption{Electronic components in the experiment}
\centering
		\begin{tabular}{ | l | l | l | p{1.5cm} |}
			\hline
			Symbol & Parameter  & Value & Units  \\ \hline 
			$R_i $ & resistor & $1.0\pm5\%$ & $K \Omega$ \\ \hline
			$R_j $ & resistor & $1.0\pm5\%$ & $M \Omega$ \\ \hline
			$C_i $ & capacitor & $1.0\pm10\%$ & $ \mu F$ \\ \hline
			op-amp  & UA741CN  &  &  \\ \hline
			Potentiometer & variable resistor & $100.0\pm5\%$ & $K \Omega$ \\ \hline
			$V_1$ and $V_2$ & voltage source  & $\pm 12$ & V  \\ \hline
			Analog multiplier & AD633JNZ & &\\ \hline
		\end{tabular}
        \label{Table_1}
\end{table}

\textit{Example: Van der Pol oscillator.} In this section, we extend our investigation into the synchronization dynamics of coupled VPOs by considering a system comprising three oscillators interconnected via partial-state coupling and follows the VPO electronic circuit. The theoretical results are experimentally validated using a breadboard implementation of three partially coupled VPOs with distinct parameters (Table~1). The observed synchronization behavior aligns with numerical simulations, confirming the practical feasibility of the proposed model.

\begin{figure}[ht]
\centering
\includegraphics[width=0.8\textwidth]{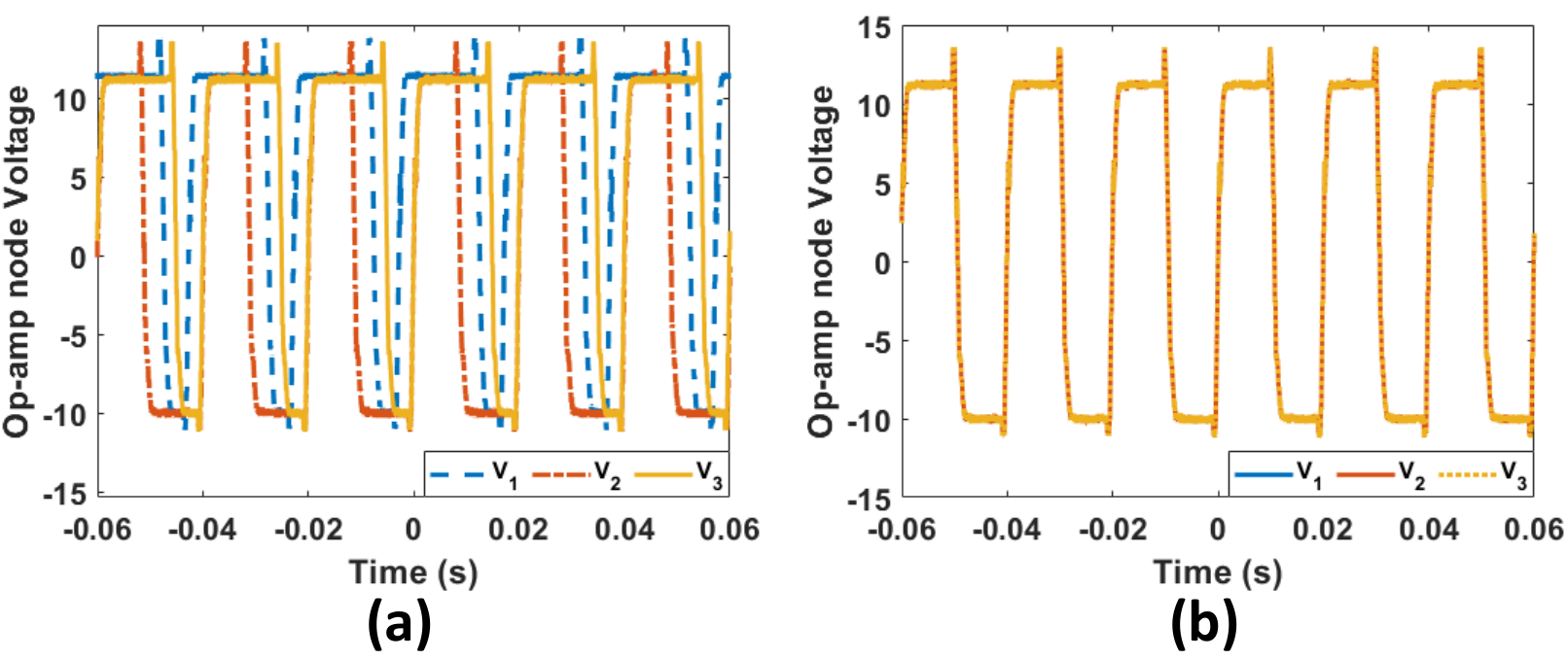}
\caption{Synchronization of three partial-state coupled VPOs. a) Unsynchronized result b) Synchronized result.}
\label{3VPO_1state_cou_exp_plot}
\end{figure}

Oscilloscope traces in Fig.~\ref{3VPO_1state_cou_exp_plot} confirm synchronization among three partially coupled VPOs with distinct parameters (Table~\ref{Table_1}). Initially unsynchronized (Fig.~\ref{3VPO_1state_cou_exp_plot}a), the oscillators achieve unified frequency and phase after coupling (Fig.~\ref{3VPO_1state_cou_exp_plot}b), validating the theoretical predictions.

\section{Conclusion}
Existing synchronization conditions are largely sufficient and conservative, despite numerical evidence suggesting tighter thresholds. Using Floquet theory, we prove that a positive coupling constant ($K>0$) is necessary and sufficient for local synchronization of identical oscillators under full-state linear coupling. For partial-state coupling, positive coupling ensures asymptotic state-space volume contraction, implying local synchronization for two-dimensional oscillators. We further show that positive coupling gains remain sufficient for local synchronization under non-identical coupling over directed graphs. Numerical simulations and experiments with coupled electronic oscillators validate the results. These findings provide a sharper characterization of local synchronization and reduce conservatism in existing theory. Future work will address heterogeneous oscillators, higher-dimensional partial coupling, and switching directed topologies, with applications in power systems, multi-agent coordination, neuromorphic networks, and distributed clock synchronization.

\bibliographystyle{ieeetr}
\bibliography{reference}

@book{izhikevich2007dynamical,
  title={Dynamical Systems in Neuroscience},
  author={Izhikevich, E.M.},
  isbn={9780262090438},
  lccn={2006040349},
  series={Computational neuroscience Dynamical systems in neuroscience},
  year={2007},
  address		= "55 Hayward St, Cambridge, MA 02142, United States",
  publisher={MIT Press}
}

@incollection{pikovsky2003synchronization,
  title={Synchronization: a universal concept in nonlinear sciences},
  author={Pikovsky, Arkady and Kurths, Jurgen and Rosenblum, Michael and Kurths, J{\"u}rgen},
  number={12},
  year={2003},
  publisher={Cambridge university press}
}

@book{brockett2015finite,
  title={Finite Dimensional Linear Systems},
  author={Brockett, R.W.},
  isbn={9781611973877},
  lccn={74100326},
  series={Classics in Applied Mathematics},
  address		= "Philadelphia",
  year={2015},
  publisher={Society for Industrial and Applied Mathematics}
}

@inproceedings{scardovi2008synchronization,
  title={Synchronization in networks of identical linear systems},
  author={Scardovi, Luca and Sepulchre, Rodolphe},
  booktitle={ 47th Conference on Decision and Control (CDC)},
  pages={546--551},
  year={2008},
  organization={IEEE}
}

@book{iakubovich_linear_1975,
	address = {New York},
	title = {Linear differential equations with periodic coefficients},
	isbn = {9780470969533},
	language = {engrus},
	publisher = {Wiley},
	author = {Yakubovich, V. A. and Starzhinskii, V. M.},
	year = {1975},
}

@inproceedings{jadbabaie2004stability,
  title={On the stability of the Kuramoto model of coupled nonlinear oscillators},
  author={Jadbabaie, Ali and Motee, Nader and Barahona, Mauricio},
  booktitle={Proceedings of the 2004 American Control Conference},
  volume={5},
  pages={4296--4301},
  year={2004},
  organization={IEEE}
}

@article{barahona2002synchronization,
  title={Synchronization in small-world systems},
  author={Barahona, Mauricio and Pecora, Louis M},
  journal={Physical review letters},
  volume={89},
  number={5},
  pages={054101},           
  year={2002},
  publisher={APS}
}

@article{sinha1996liapunov,
author = {Sinha, S. and Pandiyan, R. and Bibb, J.},
year = {1996},
pages={209--219},
title = {Liapunov-Floquet Transformation: Computation and Applications to Periodic Systems},
volume = {118},
journal = {Journal of Vibration and Acoustics},
doi = {10.1115/1.2889651}
}

@article{carroll1993synchronizing,
  title={Synchronizing nonautonomous chaotic circuits},
  author={Carroll, Thomas L and Pecora, Louis M},
  journal={IEEE Transactions on Circuits and Systems II: Analog and Digital Signal Processing},
  volume={40},
  number={10},
  pages={646--650},
  year={1993},
  publisher={IEEE}
}

@article{pecora1998master,
  title={Master stability functions for synchronized coupled systems},
  author={Pecora, Louis M and Carroll, Thomas L},
  journal={Physical Review Letters},
  volume={80},
  number={10},
  pages={2109--2112},
  year={1998},
  publisher={APS}
}

@article{pecora1990synchronization,
  title={Synchronization in chaotic systems},
  author={Pecora, Louis M and Carroll, Thomas L},
  journal={Physical review letters},
  volume={64},
  number={8},
  pages={821},
  year={1990},
  publisher={APS}
}

@article{buldu2007electronic,
  title={Electronic design of synthetic genetic networks},
  author={Buld{\'u}, Javier M and Garc{\'\i}a-Ojalvo, Jordi and Wagemakers, Alexandre and Sanju{\'a}n, Miguel AF},
  journal={International Journal of Bifurcation and Chaos},
  volume={17},
  number={10},
  pages={3507--3511},
  year={2007},
  publisher={World Scientific}
}

@article{ermentrout1991adaptive,
  title={An adaptive model for synchrony in the firefly Pteroptyx malaccae},
  author={Ermentrout, Bard},
  journal={Journal of Mathematical Biology},
  volume={29},
  number={6},
  pages={571--585},
  year={1991},
  publisher={Springer}
}

@article{arcak2011certifying,
  title={Certifying spatially uniform behavior in reaction--diffusion PDE and compartmental ODE systems},
  author={Arcak, Murat},
  journal={Automatica},
  volume={47},
  number={6},
  pages={1219--1229},
  year={2011},
  publisher={Elsevier}
}

@article{wu1995synchronization,
  title={Synchronization in an array of linearly coupled dynamical systems},
  author={Wu, Chai Wah and Chua, Leon O},
  journal={IEEE Transactions on Circuits and Systems I: Fundamental Theory and Applications},
  volume={42},
  number={8},
  pages={430--447},
  year={1995},
  publisher={IEEE}
}

@article{joshi2022synchronization,
  title={Synchronization of coupled benchmark oscillators: analysis and experiments},
  author={Joshi, Shyam Krishan and Sen, Shaunak and Kar, Indra Narayan},
  journal={International Journal of Dynamics and Control},
  volume={10},
  number={2},
  pages={577--597},
  year={2022},
  publisher={Springer}
}

@article{shafi2013synchronization,
  title={Synchronization of diffusively-coupled limit cycle oscillators},
  author={Shafi, S Yusef and Arcak, Murat and Jovanovi{\'c}, Mihailo and Packard, Andrew K},
  journal={Automatica},
  volume={49},
  number={12},
  pages={3613--3622},
  year={2013},
  publisher={Elsevier}
}

@article{strogatz1996nonlinear,
  title={Nonlinear dynamics and chaos},
  author={Strogatz, Steven H},
  year={1996}
}

@article{russo2010global,
  title={Global convergence of quorum-sensing networks},
  author={Russo, Giovanni and Slotine, Jean Jacques E},
  journal={Physical Review E},
  volume={82},
  number={4},
  pages={041919},
  year={2010},
  publisher={APS}
}

@article{strogatz2004sync,
  title={Sync: The emerging science of spontaneous order},
  author={Strogatz, Steven},
  year={2004},
  publisher={Penguin UK}
}

@article{brody2003simple,
  title={Simple networks for spike-timing-based computation, with application to olfactory processing},
  author={Brody, Carlos D and Hopfield, JJ},
  journal={Neuron},
  volume={37},
  number={5},
  pages={843--852},
  year={2003},
  publisher={Elsevier}
}

@article{singer1999neuronal,
  title={Neuronal synchrony: a versatile code for the definition of relations?},
  author={Singer, Wolf},
  journal={Neuron},
  volume={24},
  number={1},
  pages={49--65},
  year={1999},
  publisher={Elsevier}
}

@article{tang1983synchronization,
  title={Synchronization and chaos},
  author={Tang, YS and Mees, A and Chua, L},
  journal={IEEE Transactions on Circuits and Systems},
  volume={30},
  number={9},
  pages={620--626},
  year={1983},
  publisher={IEEE}
}

@article{liu2019synchronization,
  title={Synchronization control between two {C}hua's circuits via capacitive coupling},
  author={Liu, Zhilong and Ma, Jun and Zhang, Ge and Zhang, Yin},
  journal={Applied Mathematics and Computation},
  volume={360},
  pages={94--106},
  year={2019},
  publisher={Elsevier}
}

@article{kopell2000we,
  title={We got rhythm: Dynamical systems of the nervous system},
  author={Kopell, Nancy},
  journal={Notices of the AMS},
  volume={47},
  number={1},
  pages={6--16},
  year={2000}
}

@article{nunez2016synchronization,
  title={Synchronization of pulse-coupled oscillators to a global pacemaker},
  author={Nunez, Felipe and Wang, Yongqiang and Teel, Andrew R and Doyle III, Francis J},
  journal={Systems \& Control Letters},
  volume={88},
  pages={75--80},
  year={2016},
  publisher={Elsevier}
}

@article{ajala2021robust,
  title={Robust leader--follower synchronization of electric power generators},
  author={Ajala, Olaoluwapo and Dom{\'\i}nguez-Garc{\'\i}a, Alejandro D and Liberzon, Daniel},
  journal={Systems \& Control Letters},
  volume={152},
  pages={104937},
  year={2021},
  publisher={Elsevier}
}

@article{sinha2023coupled,
  title={Coupled output synchronization in networked systems},
  author={Sinha, Pallavi and Dutta, Maitreyee and Srikant, Sukumar},
  journal={Systems \& Control Letters},
  volume={181},
  pages={105646},
  year={2023},
  publisher={Elsevier}
}

@article{braga2024selecting,
  title={Selecting the coupling variable to synchronize nonlinear oscillators},
  author={Braga, Pedro Augusto da Silva and Aguirre, Luis Antonio},
  journal={Nonlinear Dynamics},
  volume={112},
  number={17},
  pages={15177--15191},
  year={2024},
  publisher={Springer}
}

@article{boccaletti2006complex,
  title={Complex networks: Structure and dynamics},
  author={Boccaletti, Stefano and Latora, Vito and Moreno, Yamir and Chavez, Martin and Hwang, D-U},
  journal={Physics reports},
  volume={424},
  number={4-5},
  pages={175--308},
  year={2006},
  publisher={Elsevier}
}

\end{document}